\documentclass[conference]{IEEEtran}
\IEEEoverridecommandlockouts
\ifCLASSINFOpdf
\else
\fi
\usepackage[cmex10]{amsmath}
\usepackage{amsthm}
\usepackage{amssymb}
\usepackage{mathrsfs}
\usepackage[dvips]{graphicx}
\newtheorem{remark}{Remark}
\newtheorem{theorem}{Theorem}
\begin{document}

%\thanks{This work was partially supported by Iranian National Science Foundation
%(INSF) under Contract No. 88114.46 and by Iran Telecom Research Center (ITRC).}

\title{Achievable Rate Region for Multiple Access Channel with Correlated Channel States and Cooperating Encoders}

%\author{\IEEEauthorblockN{Mahdi Zamanighomi}
%\IEEEauthorblockA{School of Electrical and\\Computer Engineering\\
%%Georgia Institute of Technology\\
%%Atlanta, Georgia 30332--0250\\
%%Email: http://www.michaelshell.org/contact.html}
%\and
%\IEEEauthorblockN{Mohammad Javad Emadi}
%%\IEEEauthorblockA{Twentieth Century Fox\\
%%Springfield, USA\\
%%Email: homer@thesimpsons.com}
%\and
%\IEEEauthorblockN{F. Shirani}
%\IEEEauthorblockA{%Starfleet Academy\\
%%San Francisco, California 96678-2391\\
%%Telephone: (800) 555--1212\\
%%Fax: (888) 555--1212}
%}
%\and
%\IEEEauthorblockN{M. R. Aref}
%\IEEEauthorblockA{%Twentieth Century Fox\\
%%Springfield, USA\\
%%Email: homer@thesimpsons.com}
%}}
\author{\IEEEauthorblockN{Mahdi Zamanighomi, Mohammad Javad Emadi, Farhad Shirani Chaharsooghi, and Mohammad Reza Aref }
\IEEEauthorblockA{Information Systems and Security Lab (ISSL)\\
Electrical Engineering Department, Sharif University of Technology, Tehran, Iran.\\
Email: m\_zamani@ee.sharif.edu, emadi@ee.sharif.edu, fshirani@ee.sharif.edu, aref@sharif.edu}}


%Telephone: (800) 555--1212, Fax: (888) 555--1212}


\maketitle

\begin{abstract}
%\boldmath
In this paper, a two-user discrete memoryless multiple-access channel (DM-MAC) with correlated channel states, each known at one of the encoders is considered, in which each encoder transmits independent messages and tries to cooperate with the other one. To consider cooperating encoders, it is assumed that each encoder strictly-causally receives and learns the other encoder's transmitted symbols and tries to cooperate with the other encoder by transmitting its message. Next, we study this channel in a special case; we assume that the common part of both states is known at both, hence encoders use this opportunity to get better rate region. For these scenarios, an achievable rate region is derived based on a combination of block-Markov encoding and Gel'fand-Pinsker coding techniques. Furthermore, the achievable rate region is established for the Gaussian channel, and it is shown that the capacity region is achieved in certain circumstances.
\end{abstract}


% \fi
\IEEEpeerreviewmaketitle

\section{\textbf{Introduction}}
Channels with states have become very important in communication especially in information theory \cite{4}. Initially, Shannon \cite{1} proposed single user discrete memoryless state-dependent channel with causal channel state information at transmitter (CSIT) and evaluated the capacity of this channel.  Subsequently, Gel'fand and Pinkser \cite{2} extended this type of problem to a channel with non-causal CSIT and characterized the capacity of the channel. The next stage was completed by Costa \cite{3}, who attained the Gaussian channel in Gel'fand-Pinsker's channel and demonstrated that the dirty paper coding (DPC) completely mitigates the effect of Gaussian additive interference, which is known non-causally at the transmitter. As a result of a growing range of applications, studying a multi-user model with random parameters has received considerable attention. To review more related studies, refer to \cite{4} , \cite{5}. In this paper, we focus on the multiple-access channel considered in (MAC) \cite{6} and for this channel we make use of the scenarios considered in \cite{7}-\cite{11}, especially for the Gaussian channel. However, in the Gaussian channel, we use generalized dirty paper coding (GDPC) used in \cite{12}.

Moreover, cooperation in MAC is also important. Willems in \cite{13} explained the MAC with partially cooperating encoders, which needs a noise-free limited-rate links to set up cooperation between two encoders. In \cite{14}-\cite{17}, various models of state-dependent MAC with conferencing links were used for this purpose. On the other hand, Willems and Van der Meulen have also considered the MAC with different strategies of cribbing encoders in \cite{18}, which takes advantage of the nature of wireless networks. Hence, in this way, there is no need to allocate conferencing links between two encoders and cooperation between encoders is implemented via cribbing each other's transmitted signals. Recently, in \cite{19}, the two-user discrete memoryless MAC (DM-MAC), in which encoder 2 cribs causally or strictly-causally from encoder 1, and encoder 2 knows CSIT non-causally, has been studied and the capacity region has been established.

In this paper, we study a more practical model in which encoder 1 strictly-causally receives and learns the transmitted channel inputs of encoder 2, and encoder 2 strictly-causally receives and learns the transmitted channel inputs of encoder 1. We have also assumed that both encoders have non-causal correlated CSITs. To consider the correlated CSITs, we have considered two scenarios: first, $\textbf{s}_\textbf{1}$ and $\textbf{s}_\textbf{2}$ are considered such that they are correlated and non-causally available at encoders 1 and 2 respectively; second, three independent states, $\textbf{s}_\textbf{0}$, $\textbf{s}_\textbf{1}$ and $\textbf{s}_\textbf{2}$ are considered and the pairs ($\textbf{s}_\textbf{0}$,$\textbf{s}_\textbf{1}$) and ($\textbf{s}_\textbf{0}$,$\textbf{s}_\textbf{2}$) are non-causally available at encoders 1 and 2, respectively. For both scenarios an achievable rate regions are then established. The Gaussian channel is considered and an the achievable rate region is characterized. Finally, we manage to achieve the capacity region for the Gaussian channel under certain conditions.

The remainder of the paper is organized as follows. In Section II, the channel model under consideration is introduced. In Section III, the main results are presented for the DM-MAC. In Section IV, the Gaussian channel is considered and the achievable rate region is established. Finally, the paper is concluded in Section VI.

%I wish you the best of success.

%\hfill

%\hfill
%%%%%%%%%%%%%%%%%%%%%%%%%%%%%%%%%%%%%%%%%%%%%%%%%%%%%%%%%%%%%%%%%%%%%%%%%%%%%%%%%%%%%%%%%%%%%%%%%%%%%%%%%%%%%%%%%%%%%%%%%%%%%%%%%%%%%%%%%%%%%%%%%%%%

\section{\textbf{Channel Model}}
The following notation is used throughout the paper: random variables (r.vs), and their realizations are denoted by capital letters and lower case letters, respectively. The bold face notation $\textbf{x}$ is used to show n-vector $\textbf{x}=(x_1,x_2,...,x_n)$. The probability density function (pdf) of r.v $X$ will be denoted by $P_X$, and $P_{Y|X}$ stands for conditional pdf of  $Y$ given $X$.
In our model, we have considered a two-user discrete memoryless state-dependent multiple-access channel (Fig.1), defined by a triple $(\mathcal{X}_1\times\mathcal{X}_2\times\mathcal{S}_1\times\mathcal{S}_2,P_{Y|X_1,X_2,S_1,S_2},\mathcal{Y})$ where $\mathcal{X}_i$, $\mathcal{S}_i$ for $i=1,2$, and $\mathcal{Y}$ are the finite sets alphabets of inputs, states, and an output, respectively. We have assumed the two r.vs $S_1$ and $S_2$ are correlated according to the pdf $P_{S_1 S_2} (s_1,s_2)$ and only $\textbf{s}_\textbf{i}$, $i=1,2$ is non-causally known to the \emph{i}th encoder. We have also assumed that as encoder 1 sends its own symbols, the other encoder receives the symbols of encoder 1 strictly-causally and vice versa. In wireless networks, when one node transmits its own information, it is possible for the neighbor-nodes to receive its signal. Therefore, neighbors help the node to transmit the message. Hence, by considering this concept, the cooperation between encoders in our model is based on the nature of wireless networks and it needs no extra dedicated links between two encoders. This model without channel states-dependent was introduced and its capacity region was also established for the DM-MAC by Willems-Van der Meulen \cite{18}. In our model, the channel and state processes are memoryless, i.e.,
\begin{eqnarray}
\nonumber &P(\textbf{y}|\textbf{x}_\textbf{1},\textbf{x}_\textbf{2},\textbf{s}_\textbf{1},\textbf{s}_\textbf{2})=\prod_{i=1}^n P(y_i|x_{1i},x_{2i},s_{1i},s_{2i})\\\nonumber
&P(\textbf{s}_\textbf{1},\textbf{s}_\textbf{2})=\prod_{i=1}^n P(s_{1i},s_{2i})
\end{eqnarray}
%%%%%%%%%%%%%%%%%%%%%%%%%%%%%%%%%%%%%%%%%%%%%%%%%%%%%%%%%%%%%%%%%%%%%%%%%%%%%%%%%%%%%%%%%%%%%%%%%%%%%%%%%%%%%%%%%%%%%%%%%%%%%%%%%%%%%%%%%%%%%%%%%%
\textbf{Encoding and decoding}: In the two-user state-dependent MAC with cooperating encoders shown in Fig.1, the \emph{i}th encoder sends a message $W_{i}$, which is drawn uniformly from the set $W_{i}\triangleq\{1,2,...,2^{nR_i}\}$ to the receiver for $i=1,2$. it is also assumed that $W_1$ and $W_2$ are independent. A length-n code $C^n (R_1,R_2)$ consists of two encoding function sets ${f_{ij} (.)}_{j=1}^n$, $i=1,2$ and a decoding function $g(.)$ as follows:
\begin{eqnarray}
\left\{
\begin{array}{rl}
f_{1,j}: \mathcal{W}_1 \times \mathcal{S}_1^n \times \mathcal{X}_2^{j-1} \rightarrow \mathcal{X}_1 & \text{for }  j=1,...,n\\
f_{2,j}: \mathcal{W}_2 \times \mathcal{S}_2^n \times \mathcal{X}_1^{j-1} \rightarrow \mathcal{X}_2 & \text{for }  j=1,...,n
\end{array} \right.
\end{eqnarray}

As a result, we can write $X_{1,j}=f_{1,j} (W_1,S_1^n,X_2^{j-1})$ and $X_{2,j}=f_{2,j} (W_2,S_2^n,X_1^{j-1})$ and the decoding function is given by
\begin{equation}
g(.): \mathcal{Y}^n \rightarrow \mathcal{W}_1 \times \mathcal{W}_2
\end{equation}
which estimates the transmitted messages $W_1$  and $W_2$  as $(W_1,W_2)=g(Y^n)$. For a given code, the average of error probability is
\begin{eqnarray}
&P_e^n=\frac{1}{2^{n(R_1+R_2)}}\sum_{w_1,w_2}Pr\biggl\{\\\nonumber
&(\hat{W_1},\hat{W_2})\neq(w_1,w_2)|(w_1,w_2)\  has\  been\  sent\ \biggr\}
\end{eqnarray}

A rate-pair $(R_1,R_2)$ is said to be achievable if there exists a sequence of length-n code $C^n (R_1,R_2)$ with $P_e^n\rightarrow0$ as $n\rightarrow\infty$. The channel capacity region is the closure of all achievable rates.
\begin{flushleft}
\begin{figure}[ut]
%\centering
\includegraphics[width = 3.5in]{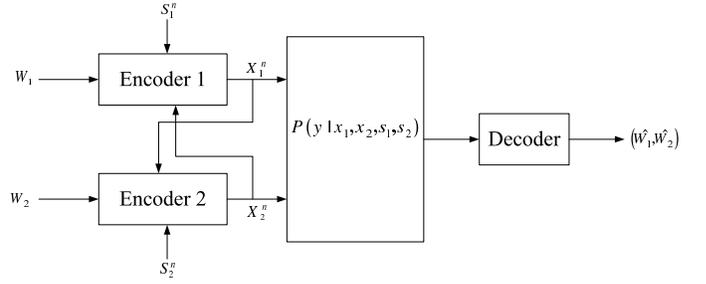}
\caption{the two-user state-dependent MAC Channel with correlated states and cooperating encoders}
\end{figure}
\end{flushleft}
%%%%%%%%%%%%%%%%%%%%%%%%%%%%%%%%%%%%%%%%%%%%%%%%%%%%%%%%%%%%%%%%%%%%%%%%%%%%%%%%%%%%%%%%%%%%%%%%%%%%%%%%%%%%%%%%%%%%%%%%%%%%%%%%%%%%%%%%%%	 

\section{\textbf{Main Result}}
In this section, we state and prove our main results for the two-user state-dependent DM-MAC with cooperating encoders. In the first scenario, it is assumed that the two correlated states, $\textbf{s}_\textbf{1}$ and $\textbf{s}_\textbf{2}$, are non-causally known at encoders 1 and 2, respectively, but both encoders have no idea which part of their states is correlated. In the second scenario, it is assumed that the channel state is a triple $(\textbf{s}_\textbf{0},\textbf{s}_\textbf{1},\textbf{s}_\textbf{2})$ where $\textbf{s}_\textbf{0},\textbf{s}_\textbf{1},$ and $\textbf{s}_\textbf{2})$ are independent. It is also considered that encoder 1 knows $(\textbf{s}_\textbf{0},\textbf{s}_\textbf{1})$ and encoder 2 knows $(\textbf{s}_\textbf{0},\textbf{s}_\textbf{2})$. In the following, achievable rate regions for two scenarios are established.
\begin{theorem}
 For the two-user DM-MAC with cooperating encoders and two correlated states, $\textbf{s}_\textbf{1}$ and $\textbf{s}_\textbf{2}$, known non-causally to encoders 1 and 2, respectively, the following rate region is achievable:
\begin{eqnarray}
\nonumber&(R_1,R_2)=\bigcup\\
&\left\{
\begin{array}{rl}
R_1\leq I(X_1;V_1|U,S_2)-I(V_1;S_1|U,S_2)\\
R_2\leq I(X_2;V_2|U,S_1)-I(V_2;S_2|U,S_1)\\
R_1+R_2\leq I(Y;V_1,V_2,U)-I(V_1,V_2;S_1,S_2|U)
\end{array} \right\}
\end{eqnarray}
where $U,V_1$, and $V_2$ are three auxiliary r.vs with finite alphabets $\mathcal{U},\mathcal{V}_1$, and $\mathcal{V}_2$, respectively and the union is over the joint pdf of r.vs $S_1,S_2,U,V_1,V_2,X_1,X_2,$ and $Y$ as follows,
\begin{eqnarray}
&\nonumber P_{S_1 S_2 UV_1 V_2 X_1 X_2 Y} (s_1,s_2,u,v_1,v_2,x_1,x_2,y)=\\\nonumber
&P_{S_1 S_2} (s_1,s_2) P_U (u) P_{X_1 V_1 |US_1} (x_1,v_1|u,s_1)\times\\
&P_{X_2 V_2 |US_2} (x_2,v_2|u,s_2) P_{Y|X_1 X_2 S_1 S_2} (y|x_1,x_2,s_1,s_2)
\label{5}
\end{eqnarray}
\end{theorem}
\begin{remark}Based on Markov chain $S_{3-i}-S_iU-V_i$, $I(X_i;V_i|U,S_{3-i})-I(V_i;S_i|U,S_{3-i})=I(X_i,S_{3-i};V_i|U)-I(V_i;S_i|U), \forall i=1,2$.
\end{remark}
\begin{theorem}
 For the two-user DM-MAC with cooperating encoders and three independent states $\textbf{s}_\textbf{0}$, $\textbf{s}_\textbf{1}$  and $\textbf{s}_\textbf{2}$, such that $(\textbf{s}_\textbf{0},\textbf{s}_\textbf{1})$ and $(\textbf{s}_\textbf{0},\textbf{s}_\textbf{2})$ are known non-causally to encoders 1 and 2, respectively, the following rate region is achievable:
\begin{eqnarray}
\nonumber&(R_1,R_2)=\bigcup\\
&\left\{
\begin{array}{rl}
R_1\leq I(X_1;V_1|U,S_0)-I(V_1;S_1|U,S_0)\\
R_2\leq I(X_2;V_2|U,S_0)-I(V_2;S_2|U,S_0)\\
R_1+R_2\leq I(Y;V_1,V_2,U)-I(U,V_1,V_2;S_0,S_1,S_2)
\end{array} \right\}
\label{6}
\end{eqnarray}
where $U,V_1,$ and $V_2$ are three auxiliary r.vs with finite alphabets $\mathcal{U},\mathcal{V}_1$, and $\mathcal{V}_2$, respectively and the union is over the joint pdf of r.vs $S_0,S_1,S_2,U,V_1,V_2,X_1,X_2,$ and $Y$ as,
\begin{align}
&\nonumber\ \ \ \ \ \ \ \ P_{S_0 S_1 S_2 U V_1 V_2 X_1 X_2 Y} (s_0,s_1,s_2,u,v_1,v_2,y)=\\\nonumber
&\ \ \ \ \ \ \ \ \ \ \ \  P_{S_0} (s_0) P_{S_1} (s_1) P_{S_2} (s_2) P_{U|S_0} (u|s_0)\times\\\nonumber
&P_{X_1 V_1 |U S_0 S_1} (x_1,v_1|u,s_0,s_1) P_{X_2 V_2 |U S_0 S_2 } (x_2,v_2|u,s_0,s_2)\times\\
&\ \ \ \ \ \ \ \ \ \ \ \ \ \ \ P_{Y|X_1 X_2 S_0 S_1 S_2} (y|x_1,x_2,s_0,s_1,s_2)
\label{7}
\end{align}
\end{theorem}
\begin{proof}[Proof of Theorem 1] In the achievability proof, we use block-Markov encoding (BME) \cite{20}-\cite{21} and Gel'fand-Pinsker coding (GPC) techniques. As each encoder receives the transmitted symbols of the other, it tries to decode the other's message and transmit it in the next block with its new information. Thus, each encoder uses superposition BME for the transmission of new and old information, and both encoders cooperate with each other by transmitting other's message. We consider $B$ blocks and each block, has $n$ symbols. The encoders send a certain amount of new and old information in each block. In the block $b$, the fresh information is a message pair $(W_1^b,W_2^b)$ and old information is $W_0^{b}=(W_1^{b-1},W_2^{b-1})$ for $b=1,2,...,B-1$. $W_1^b$ and $W_2^b$ are independent and identically distributed (i.i.d) sequences with a uniform distribution over $\{1,2,...,2^{nR_1}\}$ and $\{1,2,...,2^{nR_2}\}$. We have $B-1$ blocks which contain new information. If we assume that $B$ tends to be infinite, for a fixed $n$, then $\frac{R_1(B-1)}{B}\rightarrow R_1$ and $\frac{R_2(B-1)}{B}\rightarrow R_2$.
%We also take into account three random variables, $U\in U,$ $V_1\in V_1$, and $V_2\in V_2$, which have been considered in \emph{theorems 1} and 2, where $V_1$  and $V_2$  are chosen based on the GPC.

\textbf{Codebook generation}: In each block, the same codebooks are generated as follows:
1) Generate $2^{n(R_1+R_2)}$ sequences  $\textbf{u}=(u_1,u_2,...,u_n)$, each with the probability $Pr(\textbf{u})=\prod_{i=1}^n P(u_i)$ .Label them $\textbf{u}(w_0)$for $w_0\in \{1,2,...,2^{n(R_1+R_2)}\}$.

2) For each $\textbf{u}(w_0)$, generate $2^{n(R_1+R_{s_1 })}$ sequences $\textbf{v}_\textbf{1}=(v_{11},v_{12},...,v_{1n})$ according to $Pr(\textbf{v}_\textbf{1}|\textbf{u})=\prod_{i=1}^n P(v_{1i}|u_{i(w_0 )})$. Label them $\textbf{v}_\textbf{1}(w_0,w_1,j)$, whereby $w_1$ denotes the bin and  $j$ identifies the index in the $w_1$th bin. Also, $j\in \{1,2,...,2^{nR_{s_1 }}\}$ and $w_1\in \{1,2,...,2^{nR_1}\}$.

3) Similarly, for each $\textbf{u}(w_0 )$ generate $2^{n(R_2+R_{s_2})}$ sequences $\textbf{v}_\textbf{2}=(v_{21},v_{22},...,v_{2n})$ with the probability $Pr(\textbf{v}_\textbf{2}|\textbf{u})=\prod_{i=1}^n P(v_{2i}|u_{i(w_0 )})$ . Label them $\textbf{v}_\textbf{2}(w_0,w_2,k)$, whereby $w_2$ identifies the bin and $k$ is the index in the $w_2$th bin. Also, $k\in \{1,2,...,2^{nR_{s_2}}\}$ and $w_2\in \{1,2,...,2^{nR_2}\}$.
% %%%%%%%%%%%%fiure2 should be here%%%%%%%%%%
 %To clarify the usage of three codebooks $\textbf{u},\textbf{v}_\textbf{1}$ and $\textbf{v}_\textbf{2}$, see Fig. 2.
 %as illustrated in Fig.2
 We use the BME, since each encoder shares its own information with the other encoder and this process needs memory. Therefore, for each encoder, two variables are needed; $\textbf{u}(w_0^{b})$ denotes the shared message (old information) which is common to both encoders to set up cooperation and $\textbf{v}_\textbf{1}(w_0^b,w_1^b,j)$, and $\textbf{v}_\textbf{2}(w_0^b,w_2^b,k)$ contain both old and new information for each encoder. Since encoders 1 and 2 do not know a common part of their CSITs, $\textbf{u}$ cannot be chosen by the GPC, but $\textbf{v}_\textbf{1}$ and $\textbf{v}_\textbf{2}$ are chosen by the GPC technique based on their known CSITs. In the following, the superscript $b$ shows the index of block.

 %So, $\textbf{u}(w_0^{b})$ is used for old information to setup cooperation between encoders, and $\textbf{v}_\textbf{1}(w_0^b,w_1^b,j)$, and $\textbf{v}_\textbf{2}(w_0^b,w_2^b,k)$ contain both old and new information for each encoder. Since, encoder 1 and encoder 2 do not knowe their common CSIT, $\textbf{u}$ cannot be choosen by GPC but $\textbf{v}_\textbf{1}$ and $\textbf{v}_\textbf{2}$ are chosen by GPC technique based on their known CSITs. In the following, the superscript $b$ shows index of block.

\textbf{Encoding}: At block $b$, we denote $w_1,\text{ }w_2,\text{ }\textbf{s}_\textbf{1}=(s_{11},s_{12},...,s_{1n})$, and $\textbf{s}_\textbf{2}=(s_{21},s_{22},...,s_{2n})$  by $w_1^b,\text{ }w_2^b,\text{ }\textbf{s}_\textbf{1}^\textbf{b}=(s_{11}^b,s_{12}^b,...,s_{1n}^b)$, and $\textbf{s}_\textbf{2}^\textbf{b}=(s_{21}^b,s_{22}^b,...,s_{2n}^b)$. We use a combination of the BME and GPC techniques, in which the message pair $(w_1^b,w_2^b)$ is encoded completely in blocks $b$ and $b+1$. The transmitted signals at each encoder is defined as follows:\\
It is assumed that in block 1 encoders send
\begin{eqnarray}
&\nonumber \textbf{x}_\textbf{1}^\textbf{1}=\textbf{x}_\textbf{1} ((1,1),w_1^1,\textbf{s}_\textbf{1}^\textbf{1})\\\nonumber
&\textbf{x}_\textbf{2}^\textbf{1}=\textbf{x}_\textbf{2} ((1,1),w_2^1,\textbf{s}_\textbf{2}^\textbf{1})
\end{eqnarray}

In blocks $b=2,...,B-1$, encoders send
\begin{eqnarray}
&\nonumber \textbf{x}_\textbf{1}^\textbf{b}=\textbf{x}_\textbf{1}(w_0^{'b},w_1^b,\textbf{s}_\textbf{1}^\textbf{b} )   ,\text{  }     w_0^{'b}=(w_1^{b-1},w_2^{'b-1})\\\nonumber
&\textbf{x}_\textbf{2}^\textbf{b}=\textbf{x}_\textbf{2}(w_0^{''b},w_2^b,\textbf{s}_\textbf{2}^\textbf{b})  ,\text{  }     w_0^{''b}=(w_1^{''b-1},w_2^{b-1})
\end{eqnarray}

Finally in block $B$ we have
\begin{eqnarray}
&\nonumber \textbf{x}_\textbf{1}^\textbf{B}=\textbf{x}_\textbf{1}(w_0^{'b},1,\textbf{s}_\textbf{1}^\textbf{B})     ,\text{  }     w_0^{'B}=(w_1^{B-1},w_2^{'B-1})\\\nonumber
&\textbf{x}_\textbf{2}^\textbf{B}=\textbf{x}_\textbf{2}(w_0^{''B},1,\textbf{s}_\textbf{2}^\textbf{B})    ,\text{  }     w_0^{''B}=(w_1^{''B-1},w_2^{B-1})
\end{eqnarray}

\textbf{Generating }$\textbf{x}_\textbf{1}^\textbf{b}=\textbf{x}_\textbf{1} (w_0^{'b},w_1^b,\textbf{s}_\textbf{1}^\textbf{b})$ \textbf{at encoder 1}: Before the beginning of block $b=2,...,B$, it is assumed that encoder 1 has already estimated the transmitted message of encoder 2 at block $b-1$ as $w_2^{'b-1}$. Hence, encoder 1 uses $W_0^{'b}=(W_1^{b-1},W_2^{'b-1})$ in block $b$ as the common message. Then, encoder 1 chooses codeword $\textbf{u}(w_0^{'b})$. Now, based on the fresh information of block $b$, $w_1^b$, the bin number of the codeword $\textbf{v}_\textbf{1}$ is also determined, and encoder 1 searches for $\textbf{v}_\textbf{1}(w_0^{'b},w_1^b,j)$ with the smallest $j$ which is jointly typical with pair $(\textbf{u}(w_0^{'b}),\textbf{s}_\textbf{1}^\textbf{b})$ and describes $j$ with $j(w_0^{'b},w_1^b,\textbf{s}_\textbf{1}^\textbf{b})$. On the other hand, if there is no $j$ to satisfy joint typicality, an error occurs. Thus, the codeword $\textbf{x}_\textbf{1}^\textbf{b}=\textbf{x}_\textbf{1}(w_0^{'b},w_1^b,\textbf{s}_\textbf{1}^\textbf{b})$ with i.i.d. symbols is generated conditionally on the triple $(\textbf{u}(w_0^{'b}),\textbf{v}_\textbf{1}(w_0^{'b},w_1^b,j),\textbf{s}_\textbf{1}^\textbf{b})$, where the conditional law is induced by \eqref{5}.

\textbf{Generating} $\textbf{x}_\textbf{2}^\textbf{b}=\textbf{x}_\textbf{2}(w_0^{''b},w_2^b,\textbf{s}_\textbf{2}^\textbf{b})$ \textbf{at encoder 2}: It is assumed, because of receiving $\textbf{x}_\textbf{1}^\textbf{{b-1}}=\textbf{x}_\textbf{1}(w_0^{'b-1},w_1^{b-1},\textbf{s}_\textbf{1}^\textbf{{b-1}})$ by encoder 2, it is possible for encoder 2 to estimate $w_1^{b-1}$ as $w_1^{''b-1}$. Thus, encoder 2 makes $w_0^{''b}=(w_1^{''b-1},w_2^{b-1})$ in block $b$ as the common message. Similar to encoder 1, encoder 2 finds the codeword $\textbf{v}_\textbf{2}(w_0^{''b},w_2^b,k)$ with the smallest $k$ which is jointly typical with pair $(\textbf{u}(w_0^{''b}),\textbf{s}_\textbf{2}^\textbf{b})$, and then $\textbf{x}_\textbf{2}^\textbf{b}=\textbf{x}_\textbf{2} (w_0^{''b},w_2^b,\textbf{s}_\textbf{2}^\textbf{b})$ can be generated. On the other hand, if there is no $k$ to satisfy joint typicality, an error occurs.

\textbf{Decoding}: As mentioned, in addition to decoding at the receiver, each encoder has a decoding process, so we have:

1) Decoding at encoder 1: In order to construct the cooperation after block $b=1,2,...,B-1$, encoder 1 estimates $w_2^{'b}$ such that
\begin{equation}
\left(\textbf{u}(W_0^{'b}),\textbf{x}_\textbf{2}^\textbf{b},\textbf{v}_\textbf{2}(W_0^{'b},w_2^{'b},k),\textbf{s}_\textbf{1}^\textbf{b}\right)\in A_\epsilon^n (U,X_2,V_2,S_1)
\end{equation}
where $W_0^{'b}$has been estimated in the previous block, and $\textbf{x}_\textbf{2}^\textbf{b}$ has been received from the transmission of encoder 2.

2) Decoding at encoder 2: After block $b=1,2,...,B-1$, encoder 2 estimates $w_1^{''b}$ such that
\begin{equation}
\left(\textbf{u}(W_0^{''b}),\textbf{x}_\textbf{1}^\textbf{b},\textbf{v}_\textbf{1}(W_0^{''b},w_1^{''b},k),\textbf{s}_\textbf{2}^\textbf{b}\right)\in A_\epsilon^n (U,X_1,V_1,S_2)
\end{equation}
where $W_0^{''b}$ has been estimated in the previous block, and $\textbf{x}_\textbf{1}^\textbf{b}$ has been received from the transmission of encoder 1.
3) Decoding at the receiver: to utilize simultaneous decoding, the backward decoding technique is used at the decoder.
At the end of block $B$, $w_0^B$is estimated such that
\begin{equation}
\left(\textbf{u}(\hat{w}_0^B),\textbf{v}_\textbf{1}(\hat{w}_0^B,1,j),\textbf{v}_\textbf{2}(\hat{w}_0^B,1,k),\textbf{y}\right)\in A_\epsilon^n(U,V_1,V_2,Y)
\end{equation}
and $w_0^B=(\hat{w}_1^{B-1},\hat{w}_2^{B-1})$. For block $b=B-1,...,2$, we have
\begin{equation}
\left(\textbf{u}(\hat{w}_0^b),\textbf{v}_\textbf{1}(\hat{w}_0^b,\hat{W}_1^b,j),\textbf{v}_\textbf{2}(\hat{w}_0^b,\hat{W}_2^b,k),\textbf{y}\right)\in A_\epsilon^n (U,V_1,V_2,Y)
\end{equation}
where $\hat{W}_1^b$ and $\hat{W}_2^b$ have already been estimated in block $b+1$. Since $w_0^1=(1,1)$, it is unnecessary to decode at block 1.
\textbf{Error probability analysis}: The error probability is given by
\begin{eqnarray}
&&\nonumber\overline{P_e^B}=\bigcup_{b=1}^{B}\biggl\{\sum_{(\textbf{s}_\textbf{1}^\textbf{b},\textbf{s}_\textbf{2}^\textbf{b}\notin A_\epsilon^n(S_1,S_2))}Pr(\textbf{s}_\textbf{1}^\textbf{b},\textbf{s}_\textbf{2}^\textbf{b})\\\nonumber
&&+\sum_{(\textbf{s}_\textbf{1}^\textbf{b},\textbf{s}_\textbf{2}^\textbf{b}\in A_\epsilon^n (S_1,S_2))}Pr(Error|\textbf{s}_\textbf{1}^\textbf{b},\textbf{s}_\textbf{2}^\textbf{b})Pr(\textbf{s}_\textbf{1}^\textbf{b},\textbf{s}_\textbf{2}^\textbf{b})\biggr\}\leq\\
&&\varepsilon+\bigcup_{b=1}^{B}\biggl\{\sum_{(\textbf{s}_\textbf{1}^\textbf{b},\textbf{s}_\textbf{2}^\textbf{b}\in A_\epsilon^n (S_1,S_2))}Pr(Error|\textbf{s}_\textbf{1}^\textbf{b},\textbf{s}_\textbf{2}^\textbf{b})Pr(\textbf{s}_\textbf{1}^\textbf{b},\textbf{s}_\textbf{2}^\textbf{b})\biggr\}\ \ \ \ \
\end{eqnarray}

According to the second term, we define the error events for specific state sequences as follows:

For $b=1 ,... ,B$
\begin{align}
&\nonumber E_1^b:=\{\nexists j:1\leq j\leq2^{nR_{s1}}:(\textbf{v}_\textbf{1}(W_0,W_1,j),\textbf{s}_\textbf{1},\textbf{u}(W_0))\\\nonumber
&\ \ \ \ \ \in A_\epsilon^n (U,S_1,V_1)\}\\\nonumber
&E_2^b:=\{\nexists k:1\leq k\leq2^{nR_{s2}}:(\textbf{v}_\textbf{2}(W_0,W_2,j),\textbf{s}_\textbf{2},\textbf{u}(W_0))\\\nonumber
&\ \ \ \ \ \in A_\epsilon^n (U,S_2,V_2)\}\\\nonumber
&E_3^b:=\{(\textbf{u}(w_0),\textbf{v}_\textbf{1}(w_0,w_1,j),\textbf{v}_\textbf{2}(w_0,w_2,k),\textbf{s}_\textbf{1},\textbf{s}_\textbf{2})\\\nonumber
&\ \ \ \ \ \notin A_\epsilon^n(U,V_1,V_2,S_1,S_2)\}\\\nonumber
&E_4^b:=\{(\textbf{u}(w_0),\textbf{v}_\textbf{2}(w_0,w_2,k),\textbf{x}_\textbf{2},\textbf{s}_\textbf{1})\\\nonumber
&\ \ \ \ \ \notin A_\epsilon^n(U,V_2,X_2,S_1)\}\\\nonumber
&E_5^b:=\{\exists w_2^{'}\neq W_2:(\textbf{u}(W_0),\textbf{v}_\textbf{2}(W_0,w_2^{'},k),\textbf{x}_\textbf{2},\textbf{s}_\textbf{1})\\\nonumber
&\ \ \ \ \ \in A_\epsilon^n (U,V_2,X_2,S_1)\}\\\nonumber
&E_6^b:=\{(\textbf{u}(w_0),\textbf{v}_\textbf{1}(w_0,w_1,j),\textbf{x}_\textbf{1},\textbf{s}_\textbf{2})\\\nonumber
&\ \ \ \ \ \notin A_\epsilon^n(U,V_1,X_1,S_2)\}\\\nonumber
&E_7^b:=\{\exists w_1^{'}\neq W_1:(\textbf{u}(W_0 ),\textbf{v}_\textbf{1}(W_0,w_1^{'},j),\textbf{x}_\textbf{1},\textbf{s}_\textbf{2})\\\nonumber
&\ \ \ \ \ \in A_\epsilon^n(U,V_1,X_1,S_2)\}\\\nonumber
&E_8^b:=\{(\textbf{u}(w_0),\textbf{v}_\textbf{1}(w_0,w_1,j),\textbf{v}_\textbf{2}(w_0,w_2,k),\textbf{y})\\\nonumber
&\ \ \ \ \ \notin A_\epsilon^n (U,V_1,V_2,Y)\}\\\nonumber
&E_9^b:=\{\exists w_0^{'}\neq W_0:(\textbf{u}(w_0^{'}),\textbf{v}_\textbf{1}(w_0^{'},W_1,j),\textbf{v}_\textbf{2}(w_0^{'},W_2,k),\textbf{y})\\\nonumber
&\ \ \ \ \ \in A_\epsilon^n (U,V_1,V_2,Y)\}
\end{align}

Based on the error events, the error probability is upper bounded by
\begin{align}
\nonumber&\overline{P_e^B}\leq Pr\Biggl\{
\bigcup_{b=1}^{B-1}\Biggl(E_1^b\cup E_2^b\cup \left(E_3^b|\overline{E_1^b},\overline{E_2^b}\right)\cup\left(E_4^b|\overline{E_3^b}\right)\\\nonumber
&\bigcup_{w_2^{'}\neq W_2^b}\left(E_5^b |\overline{E_3^b},\overline{E_4^b}\right)\cup\left(E_6^b|\overline{E_3^b}\right)\bigcup_{w_1^{'}\neq W_1^b}\left(E_7^b|\overline{E_3^b},\overline{E_6^b}\right)\Biggr)\\\nonumber
&+\bigcup_{b=B}^2\Biggl(E_1^b\cup E_2^b\cup\left(E_3^b|\overline{E_1^b},\overline{E_2^b}\right)\cup\left(E_8^b|\overline{E_3^b }\right)\\\nonumber
&\bigcup_{w_0^{'}\neq W_0^b}\left(E_9^b|\overline{E_3^b},\overline{E_8^b }\right)\Biggr)\Biggr\}\\\nonumber
&\leq\sum_{b=1}^B\biggl\{Pr(E_1^b)+Pr(E_2^b)+Pr\left(E_3^b|\overline{E_1^b},\overline{E_2^b}\right)\\\nonumber
&+Pr\left(E_4^b|\overline{E_3^b}\right)+Pr\left(E_6^b |\overline{E_3^b}\right)+Pr\left(E_8^b|\overline{E_3^b}\right)\biggr\}\\\nonumber
&+\sum_{b=1}^{B-1}\sum_{w_2^{'}\neq W_2^b}Pr\left(E_5^b|\overline{E_3^b},\overline{E_4^b}\right)\\\nonumber
&+\sum_{b=1}^{B-1}\sum_{w_1^{'}\neq W_1^b}Pr\left(E_7^b |\overline{E_3^b},\overline{E_6^b}\right)\\\nonumber
&\leq(B-1)×\biggl\{Pr(E_1)+Pr(E_2)+Pr\left(E_3|\overline{E_1},\overline{E_2}\right)\\\nonumber
&+Pr\left(E_4|\overline{E_3}\right)+Pr\left(E_6|\overline{E_3}\right)+Pr(E_8|\overline{E_3})\\\nonumber
&+\sum_{w_1^{'}\neq W_1}Pr(E_7|\overline{E_3},\overline{E_6 })+\sum_{w_2^{'}\neq W_2}Pr(E_5|\overline{E_3},\overline{E_4})\\
&+\sum_{w_0^{'}\neq W_0}Pr(E_9|\overline{E_3},\overline{E_8})\biggr\}
\end{align}

We now bound each probability of error events. For independent $\textbf{s}_\textbf{1}$and \textbf{u}, the probability which $(\textbf{u},\textbf{s}_\textbf{1},\textbf{v}_\textbf{1})\in A_\epsilon^n$ is bounded by
\begin{align}
&\nonumber Pr((\textbf{u},\textbf{s}_\textbf{1},\textbf{v}_\textbf{1})\in A_\epsilon^n (U,S_1,V_1))=\\\nonumber
&\sum_{(\textbf{u},\textbf{s}_\textbf{1},\textbf{v}_\textbf{1})\in A_\epsilon^n(U,S_1,V_1)}P(\textbf{u})P(\textbf{s}_\textbf{1})P(\textbf{v}_\textbf{1}|\textbf{u})\\\nonumber
&\ \ \ \ \ \ \ \ \ \ \geq|A_\epsilon^n(U,S_1,V_1 )|2^{-n\left(H(U)-\epsilon\right)}2^{-n\left(H(S_1)-\epsilon\right)}2^{-n\left(H(V_1|U)-\epsilon\right)}\\\nonumber
&\ \ \ \ \ \ \ \ \ \ \geq2^{n\left(H(U,S_1,V_1)-\epsilon\right)}2^{-n\left(H(U)+H(S_1)+H(V_1|U)-3\epsilon\right)}\\
&\ \ \ \ \ \ \  \ \ \ =2^{-n\left(I(V_1;S_1|U)+4\epsilon\right)}
\end{align}

Consequently, we have
\begin{align}
&\nonumber Pr(E_1)\leq[1-2^{-n(I(V_1;S_1|U)+4\epsilon)}]^{2^{nR_{s_1}}}\\\nonumber
&\ \ \ \ \ \ \ \ \ \ {_\leq^a} \exp\left(2^{-n(R_{s_1}-I(V_1;S_1|U)-4\epsilon)}\right)\\
&\ \ \ \ \ \ \ \ \ \ \exp(-2^{-4n\epsilon})
\end{align}
where (a) comes from $1-x\geq \ln(x)$ and (b) follows since $R_{s_1}\geq  I(V_1;S_1|U)+4\epsilon$. As a result, if $R_{s_1}\geq I(V_1;S_1|U)+4\epsilon$ , $Pr(E_1)$ tends to zero as $n\rightarrow\infty$. Similarly, $Pr(E_2)$ tend to zero if $R_{s_2}\geq I(V_2;S_2|U)+4\epsilon$ as $n\rightarrow\infty$.

By considering the Markov lemma, $Pr\left(E_3^b|\overline{E_1^b},\overline{E_2^b}\right)$decays to zero, and $Pr\left(E_4|\overline{E_3}\right)$ , $Pr\left(E_6 |\overline{E_3}\right)$ ,and $Pr\left(E_8|\overline{E_3}\right)$ tend to zero according to the AEP Theorem \cite{22}-\cite{23}.

For evaluating $Pr\left(E_5|\overline{E_3},\overline{E_4}\right)$ we have
\begin{align}
&\nonumber Pr\left(E_5|\overline{E_3},\overline{E_4}\right)\leq2^{n(R_2+R_{s_2})}\times\\\nonumber
&\sum_{(\textbf{u},\textbf{v}_\textbf{2},\textbf{x}_\textbf{2},\textbf{s}_\textbf{1})\in A_\epsilon^n(UV_2X_2S_1)}P(\textbf{s}_\textbf{1})P(\textbf{u})P(\textbf{v}_\textbf{2}|\textbf{u})P(\textbf{x}_\textbf{2}|\textbf{us}_\textbf{1})\\\nonumber
&\leq2^{nR_2}2^{n\left(I(V_2;S_2|U)+4\epsilon\right)}2^{n\left(H(UV_2S_1X_2)+\epsilon\right)}\times\\\nonumber
&2^{-n\left(H(S_1)+H(U)+H(V_2|U)+H(X_2|US_1)-4\epsilon\right)}\\\nonumber
&=2^{n\left(R_2+I(V_2;S_2|U)-I(V_2;S_1|U)-I(X_2;V_2|US_1)+9\epsilon\right)}\\\nonumber
&=2^{n\left(R_2-I(X_2;V_2|US_1)+H(V_2|US_1)-H(V_2|US_2)+9\epsilon\right)}\\\nonumber
&\ {_=^c}2^{n\left(R_2-I(X_2;V_2|US_1)+H(V_2|US_1)-H(V_2|US_1S_2)+9\epsilon\right)}\\
&=2^{n\left(R_2-I(X_2;V_2|US_1)+I(V_2;S_2|US_1)+9\epsilon\right)}
\label{16}
\end{align}
where (c) comes from the Markov chain $S_1 - US_2 - V_2$. Hence, for $n\rightarrow\infty$, $Pr\left(E_5|\overline{E_3},\overline{E_4}\right)$ tends to zero if $R_2\leq I(X_2;V_2|US_1)-I(V_2;S_2|US_1)-9\epsilon$. In a similar way, for $n\rightarrow\infty$, $Pr\left(E_7|\overline{E_3},\overline{E_6}\right)$ decays to zero when $R_1\leq I(X_1;V_1|US_2)-I(V_1;S_1|US_2)-9\epsilon$.

To evaluate $Pr\left(E_9|\overline{E_3},\overline{E_8}\right)$ we have
\begin{align}
&\nonumber Pr((\textbf{u},\textbf{v}_\textbf{1},\textbf{v}_\textbf{2},\textbf{y})\in A_\epsilon^n (UV_1V_2Y))\\\nonumber
&\leq2^{n(R_1+R_2+R_{S_1}+R_{S_2})}\times\\\nonumber
&\sum_{\left((\textbf{u},\textbf{v}_\textbf{1},\textbf{v}_\textbf{2},\textbf{y})\in A_\epsilon^n(UV_1V_2Y)\right)}P(\textbf{u})P(\textbf{v}_\textbf{1}|\textbf{u})P(\textbf{v}_\textbf{2}|\textbf{u})P(\textbf{y})\\\nonumber
&\leq2^{n\left(R_1+R_2+I(V_1;S_1|U)+I(V_2;S_2|U)+8\epsilon\right)}2^{n\left(H(UV_1V_2Y)+\epsilon\right)}\times\\\nonumber
&2^{-n\left(H(U)+H(V_1|U)+H(V_2|U)+H(Y)-4\epsilon\right)}\\\nonumber
&=2^{n\left(R_1+R_2+I(V_1;S_1|U)+I(V_2;S_2|U)-I(V_2;V_1|U)-I(Y;V_1V_2U)+13\epsilon\right)}\\\nonumber
&=2^{n\left(R_1+R_2-I(Y;V_1V_2U)+H(V_1|UV_2)-H(V_1|US_1)+I(V_2;S_2|U)+13\epsilon\right)}\\\nonumber
&\ {_=^d}2^{n\left(R_1+R_2-I(Y;V_1V_2U)+H(V_1|UV_2)-H(V_1|US_1S_2V_2)\right)}\times\\\nonumber
&2^{\left(I(V_2;S_2S_1|U)+13\epsilon\right)}\\\nonumber
&=2^{n\left(R_1+R_2-I(Y;V_1V_2U)+I(V_1;S_1S_2|UV_2)+I(V_2;S_2S_1|U)+13\epsilon\right)}\\
&=2^{n\left(R_1+R_2-I(Y;V_1V_2U)+I(V_1V_2;S_1S_2|U)+13\epsilon\right)}
\label{17}
\end{align}

Where (d) follows from the Markov chain $S_2V_2- US_1- V_1$ and $S_1- US_2- V_2$. Therefore, for sufficiently large $n$, $Pr(E_9|\overline{E_3},\overline{E_8})$ tends to zero when $R_1+R_2\leq I(Y;V_1V_2U)-I(V_1V_2;S_1S_2|U)-13\epsilon$.

The achievability of a rate region is established by the combination of \eqref{16}-\eqref{17} for the law of form \eqref{5}.
\end{proof}
\begin{proof}[Proof of Theorem 2]
In this section, we assume that the channel states are $(\textbf{s}_\textbf{0},\textbf{s}_\textbf{1},\textbf{s}_\textbf{2})$, which are independent of each other. Thus, $(\textbf{s}_\textbf{0},\textbf{s}_\textbf{1})$ and $(\textbf{s}_\textbf{0},\textbf{s}_\textbf{2})$ are known non-causally at encoders 1 and 2, respectively. Since, both encoders know $\textbf{s}_\textbf{0}$, we can modify the generating codebooks of the \emph{Theorem 1} by choosing codeword $\textbf{u}$ based on the GPC technique. Therefore, the pdf of r.vs is written as \eqref{7}.

\textbf{codebook generation}: Generate $2^{n(R_{S_0}+R_1+R_2)}$  codewords $\textbf{u}(w_0,l)$, each with the probability $Pr(\textbf{u})=\prod_{i=1}^np(u_i)$. Similar to the generation of codebooks for \emph{Theorem 1}, we generate other codewords.It is necessary only to modify the error events of \emph{Theorem 1} as below:

1. There is one more error event as follow:
\begin{eqnarray}
\nonumber E_0^b:=\{\nexists l:1\leq j\leq 2^{nR_{S_0}}:(\textbf{s}_\textbf{0},\textbf{u}(W_0,l))\in A_\epsilon^n(U,S_0)\}
\end{eqnarray}

2. Replace $\textbf{s}_\textbf{1}$ with $(\textbf{s}_\textbf{0},\textbf{s}_\textbf{1})$ and $\textbf{s}_\textbf{2}$ with $(\textbf{s}_\textbf{0},\textbf{s}_\textbf{2})$.

3. Replace $\textbf{u}(.)$ with $\textbf{u}(.,l)$.

Using the covering lemma, we can show that $R_{S_0}\geq I(U;S_0)+2\epsilon$. Following the same procedure of error probability analysis as \emph{Theorem 1}, the achievable rate region as \eqref{6} is characterized.
\end{proof}
%%%%%%%%%%%%%%%%%%%%%%%%%%%%%%%%%%%%%%%%%%%%%%%%%%%%%%%%%%%%%%%%%%%%%%%%%%%%%%%%%%%%%%%%%%%%%%%%%%%%%%%%%%%%%%%%%%%%%%%%%%%%%%%%%%%%%%%%%%%%%%%%%%%%%%
\section{\textbf{Gaussian channel}}
In the Gaussian channel, we have considered the state-dependent MAC with cooperating encoders with three independent states $\textbf{s}_\textbf{0},\textbf{s}_\textbf{1}$ and $\textbf{s}_\textbf{2}$ such that $Y=X_1+X_2+S_0+S_1+S_2+Z$ where, $X_1$ and $X_2$ are channel inputs, $Z$ is the noise, and $Y$ is the output. Encoder 1 and 2 have access to the pairs $(\textbf{s}_\textbf{0},\textbf{s}_\textbf{1})$ and $(\textbf{s}_\textbf{0},\textbf{s}_\textbf{2})$ non-causally, respectively. Since two encoders know $s_0$ non-causally, both encoders use the DPC for the auxiliary random variable $U$. Thus, encoder 1 partially deletes the state $\textbf{s}_\textbf{1}$ and then utilizes the DPC to generate $V_1$ based on the $\textbf{s}_\textbf{0}$ and remainder of $\textbf{s}_\textbf{1}$ $( i.e., \textbf{s}_\textbf{1}^{'})$ \cite{12}. Similarly, encoder 2 first partially deletes $\textbf{s}_\textbf{2}$, then utilizes the DPC to generate $V_2$ based on the $\textbf{s}_\textbf{0}$ and remainder of $\textbf{s}_\textbf{2}$ ( i.e., $\textbf{s}_\textbf{2}^{'}$).

In the following, we have assumed that $\tilde{U},\tilde{V_1},\tilde{V_2},S_0,$ and $Z$ are independent Gaussian r.vs. It has also assumed that $\tilde{\tilde{V_1}}$, $\tilde{\tilde{V_2}},S_1$, and $S_2$ are independent Gaussian r.vs.
$\tilde{U},\tilde{\tilde{V_1}}$, and $\tilde{\tilde{V_2}}$ are $\aleph(0,1)$, $S_j\sim \aleph(0,Q_j )$,$\forall j=0,1,2,$ and $Z\sim \aleph(0,N)$. Also, each user has power constraint $P_i, i=1,2$, and dedicates $\overline{\eta_i}P_i$ of its total power, $P_i$, to remove $\textbf{s}_\textbf{i}$, hence $1-min\{1,\frac{Q_i}{P_i}\}\leq\eta_i\leq1$. Therefore
\begin{align}
\left\{
\begin{array}{ccc}
X_i:=\sqrt{\alpha_i\eta_iP_i}\tilde{U}+\sqrt{\tilde{\alpha_i}\eta_iP_i}\tilde{\tilde{V_i}}-\sqrt{\frac{\overline{\eta_i}P_i}{Q_i}}S_i\text{ ,}i=1,2\\
U:=\tilde{U}+\gamma_0S_0\\
V_i:=\tilde{V_i}+\gamma_{0i}S_0 \text{  },\text{  } \tilde{V_i}:=\tilde{\tilde{V_i}}+\gamma_iS_i^{'}\\
S_i^{'}:=\left(1-\sqrt{\frac{\overline{\eta_i}P_i}{Q_i}}\right)S_i\\
S_i^{'}\sim N(0,Q_i^{'}) : Q_i^{'}=(\sqrt{Q_i}-\sqrt{\overline{\eta_i}P_i})^2\\
\gamma_0=\frac{\sqrt{\alpha_1\eta_1P_1}+\sqrt{\alpha_2\eta_2P_2}}{(\sqrt{\alpha_1\eta_1P_1}+{\alpha_2\eta_2P_2})^2+\overline{\alpha_1}\eta_1P_1+\overline{\alpha_2}\eta_2P_2+Q_1^{'}+Q_2^{'}+N}\\
\gamma_{0i}=\frac{\sqrt{\overline{\alpha_i}\eta_iP_i}+\gamma_iQ_i^{'}}{(\sqrt{\alpha_1\eta_1P_1}+\sqrt{\alpha_2\eta_2P_2})^2+\overline{\alpha_1}\eta_1P_1+\overline{\alpha_2}\eta_2P_2+Q_1^{'}+Q_2^{'}+N}\text{ ,} i=1,2\\
\gamma_i=\frac{\sqrt{\overline{\alpha_i}\eta_iP_i}}{\overline{\alpha_1}\eta_1P_1+\overline{\alpha_2}\eta_2P_2+Q_{3-i}{'}+N}\text{ ,}i=1,2
\end{array} \right.
\label{18}
\end{align}
%Thus, at the decoder we receive,
%\begin{equation}
%Y=X_1+X_2+S_0+S_1+S_2+Z
%\end{equation}

\begin{theorem} Using Theorem 2, an achievable rate region for the Gaussian channel is as
\begin{align}
&(R_1,R_2)=\bigcup_{\eta_1,\eta_2,\alpha_1,\alpha_2}\\\nonumber
&\left\{
\begin{array}{ccc}
R_1\geq0 \text{ , } R_2\geq0\\
R_1+R_2\leq\frac{1}{2}\log_2\Biggl\{\left(1+\frac{(\sqrt{\alpha_1\eta_1P_1}+\sqrt{\alpha_2\eta_2P_2})^2}{\overline{\alpha_1}\eta_1P_1+\overline{\alpha_2}\eta_2P_2+Q_1^{'}+Q_2^{'}+N}\right)\times\\
\left(1+\frac{\overline{\alpha_1}\eta_1P_1}{\overline{\alpha_2}\eta_2P_2+Q_{2}{'}+N}\right)\left(1+\frac{\overline{\alpha_2}\eta_2P_2}{\overline{\alpha_1}\eta_1P_1+Q_{1}{'}+N}\right)         \Biggr\}
\end{array} \right\}
\end{align}
\end{theorem}
\begin{proof}[Proof of Theorem 3] Since the link between two encoders is noise-free (infinite capacity), two encoders decode other's messages completely, so only the sum-rate is needed to be evaluated based on \eqref{6}. By using \emph{Lemma 1}, we can prove \emph{Theorem 3}.

\emph{\textbf{Lemma 1}}: Using the Markov relations as
\begin{eqnarray}
\nonumber& S_1S_2 - S_0 - U\\\nonumber
&V_2S_2 - US_0 - V_1S_1
\end{eqnarray}

and equation \eqref{18}, it is proved that
\begin{align}
&\ \ \ \ \ \ \ \ \ \ \ \ \ \ \ \ \ \ \ \ \ \ \ \ \ \ \ I(Y;U,V_1,V_2)=\\\nonumber
&\ \ I(Y,S_0;U,V_1,V_2)=I(Y;U,V_1,V_2|S_0)+I(U,V_1,V_2;S_0)
\label{20}
\end{align}
\begin{align}
&\ \ \ \ \ \ \ \ \ \ \ \ \ \ \ \ \ \ \ \ \ \ \ \ \ I(Y;V_i|U,S_0)=\\\nonumber
&\ I(YS_i;V_i|U,S_0)=I(Y;V_i|U,S_0,S_i)+I(S_i;V_i|U,S_0)\\\nonumber
&\ \ \ \ \ \ \ \ \ \ \ \ \ \ \ \ \ \ \ \ \ \ \ \ \ \ \text{for},\ i=1,2
\label{21}
\end{align}

\emph{Proof}: In order to prove (20), we need to show that $h(U,V_1,V_2|Y)=h(U,V_1,V_2|Y,S_0)$. We rewrite $h(U,V_1,V_2|Y)$ as
\begin{eqnarray}
\nonumber&h(U,V_1,V_2|Y)=h\biggl(\tilde{U}+\gamma_0S_0\tilde{V_1}+\gamma_{01}S_0\tilde{V_2}+\\\nonumber
&\gamma_{02}S_0|(\sqrt{\alpha_1\eta_1P_1}+\sqrt{\alpha_2\eta_2P_2 })\tilde{U}+\\\nonumber
&(\sqrt{\overline{\alpha_1}\eta_1P_1})\tilde{\tilde{V_1}}+(\sqrt{\overline{\alpha_2}\eta_2P_2})\tilde{\tilde{V_2}}+S_0+S_1^{'}+S_2^{'}+Z\biggr)\\\nonumber
&=h\biggl(\varphi_0\varphi_1\varphi_2|(\sqrt{\alpha_1\eta_1P_1}+\sqrt{\alpha_2\eta_2P_2 })\tilde{U}+\\\nonumber
&(\sqrt{\overline{\alpha_1}\eta_1P_1})\tilde{\tilde{V_1}}+(\sqrt{\overline{\alpha_2}\eta_2P_2})\tilde{\tilde{V_2}}+S_0+S_1^{'}+S_2^{'}+Z\biggr)
\end{eqnarray}
where $\varphi_0,\varphi_1,$ and $\varphi_2$ are defined as
\begin{eqnarray}
\nonumber&\varphi_0 :=\tilde{U}-\gamma_0\biggl((\sqrt{\alpha_1\eta_1P_1}+\sqrt{\alpha_2\eta_2P_2 })\tilde{U}+\\\nonumber
&(\sqrt{\overline{\alpha_1}\eta_1P_1})\tilde{\tilde{V_1}}+(\sqrt{\overline{\alpha_2}\eta_2P_2})\tilde{\tilde{V_2}}+S_0+S_1^{'}+S_2^{'}+Z\biggr)\\\nonumber
&\varphi_i :=\tilde{V_i}-\gamma_{0i}\biggl((\sqrt{\alpha_1\eta_1P_1}+\sqrt{\alpha_2\eta_2P_2 })\tilde{U}+\\\nonumber
&(\sqrt{\overline{\alpha_1}\eta_1P_1})\tilde{\tilde{V_1}}+(\sqrt{\overline{\alpha_2}\eta_2P_2})\tilde{\tilde{V_2}}+S_0+S_1^{'}+S_2^{'}+Z\biggr)
\end{eqnarray}

\begin{flushleft}
\begin{figure}[ut]
%\centering
\includegraphics[width = 3.5in]{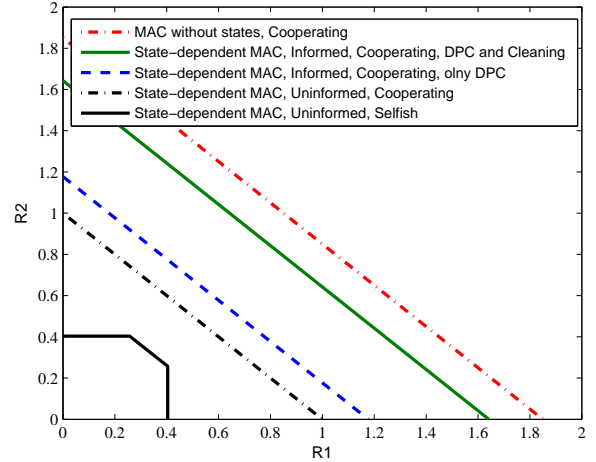}
\caption{Achievable rate region of the two-user Gaussian MAC for various scenarios. $P_1=P_2=3\times(Q_0=Q_1=Q_2=N=1)$}
\end{figure}
\end{flushleft}
We can show that $E\{\varphi_j\times Y\}=0$, $j=0,1,2$. It is also clear that $\varphi_j$ is independent of $S_0$; thus,
\begin{eqnarray}
\nonumber h(\varphi_0,\varphi_1,\varphi_2|Y)=h(\varphi_0,\varphi_1,\varphi_2)=h(\varphi_0,\varphi_1,\varphi_2|Y,S_0)
\end{eqnarray}

Following the same procedure, (21) is proved.
\end{proof}
The Gaussian channel with the channel state $(\textbf{s}_\textbf{0},\textbf{s}_\textbf{1},\textbf{s}_\textbf{2})$ and parameters $P_1=P_2=3\times(Q_0=Q_1=Q_2=N=1)$ is considered and achievable rate regions for various scenarios are shown in Fig. 3. We use "informed/uninformed" to describe whether encoders are aware/unaware of the CSIT non-causally, i.e., $(\textbf{s}_\textbf{0},\textbf{s}_\textbf{1})$ and $(\textbf{s}_\textbf{0},\textbf{s}_\textbf{2})$ are known/ not-known at encoders 1 and 2, respectively, and "cooperating /selfish" to describe whether the encoders do/do not cooperate with each other, i.e., they do/do not use the BME. As expected, when the channel is state-dependent and encoders do not have the CSIT (uninformed) and do not cooperate with each other (selfish), the achievable rate region is the smallest. However, when encoders prefer to cooperate with each other whereby they are still uninformed, it is shown that the achievable rate region becomes triangular and larger. When encoders are informed and cooperative (a combination of the BME and DPC with or without cleaning technique is used at the encoders), the achievable rate region is also triangular and is larger than the two previous scenarios. It is also clear that when encoders use the DPC and partial cleaning of their states, the achievable rate region is remarkably larger than when encoders use the DPC only. Finally, the capacity region of the model without states is plotted \cite{18}.

\begin{remark}According to the result of Theorem 3, it is clear that the effect of $S_0$ is completely removed and the channel is similar to the channel without $S_0$. Hence, when $S_1=S_2=\emptyset$, our strategy is optimum and we have achieved the capacity region which is like a triangular with boundaries determined by the lines $R_1=R_2=0$ and $R_1+R_2=\frac{1}{2}\log_2\left(1+\frac{(\sqrt{P_1}+\sqrt{P_2})^2}{N}\right)$.
\end{remark}

\begin{remark}When we have independent states $S_1$ and $S_2$, combining the DPC and cleaning techniques, might yield better performance than the DPC only. Consequently, combining the DPC and cleaning could be useful in some circumstances.
%When power constraints of two encoders are high enough, each encoder prefers to partially clean its known state $(S_i,i=1,2)$ to help other encoder and, when their powers are low, they prefer to only use DPC technique.
\end{remark}
%%%%%%%%%%%%%%%%%%%%%%%%%%%%%%%%%%%%%%%%%%%%%%%%%%%%%%%%%%%%%%%%%%%%%%%%%%%%%%%%%%%%%%%%%%%%%%%%%%%%%%%%%%%%%%%%%%%%%%%%%%%%%%%%%%%%%%%%%%%%%%%%%%%%%%%
\section{\textbf{Conclusion}}
This paper investigated the DM-MAC with two correlated states, each one of which is non-causally available for only one of the encoders. It was assumed that, each encoder receives and learns others transmitting symbols strictly-causally. Using the BME and GPC techniques, the achievable rate regions for two scenarios were established. First, general correlated states $(\textbf{s}_\textbf{1},\textbf{s}_\textbf{2})$ were considered and, second, three independent states $(\textbf{s}_\textbf{0},\textbf{s}_\textbf{1},\textbf{s}_\textbf{2})$ were considered such that pairs $(\textbf{s}_\textbf{0},\textbf{s}_\textbf{1})$ and $(\textbf{s}_\textbf{0},\textbf{s}_\textbf{2})$ were non-causally available at encoders 1 and 2, respectively. Therefore, in the second case, the encoders could choose the common auxiliary r.v, $U$, by the GPC, which causes the achievable rate region to be extended. Ultimately, the Gaussian channel for second scenario was studied and the achievable rate region using the DPC and cleaning techniques was also established. It was shown that, if each encoder partially deletes its known the CSIT before transmission, it helps the other encoder and results in a larger achievable rate region. As a result, deriving the Gaussian models showed that if both states are the same, the state's effect can be removed like Costa's DPC and the capacity region can be achieved.
\section*{Acknowledgment}
We would like to appreciate the ISSL group of Sharif
University of Technology and specially Mrs. Mirmohseni and Mrs. Akhbari for their
valuable comments.



%\bibliographystyle{IEEEtran}

\end{document}